\documentclass[12pt]{article}

\usepackage{amsmath,amsthm}
\usepackage{amssymb}
\usepackage{amsfonts}
\usepackage{amscd}
\usepackage{dsfont} 
\usepackage{mathtools} 
\usepackage{braket} 
\usepackage{color} 

\input{xy}
\xyoption{all}
\newtheorem{thm}{Theorem}[section]
\newtheorem{theorem}[thm]{Theorem}

\newtheorem{corollary}[thm]{Corollary}

\newtheorem{proposition}[thm]{Proposition}

\theoremstyle{definition}
\newtheorem{definition}[thm]{Definition}
\newtheorem{example}[thm]{Example}

\newtheorem{observation}[thm]{Observation}

\usepackage[english]{babel}
\usepackage{newlfont}
\usepackage{color}
\usepackage{hyperref}
\usepackage{multirow}
\usepackage{wrapfig}
\usepackage{cleveref}
\usepackage{caption}
\usepackage{float}
\usepackage{refcount}
\usepackage{gettitlestring}
\usepackage{csquotes}
\usepackage[dvipsnames, table]{xcolor}
\usepackage{dsfont}

\usepackage{zx-calculus}

\newcommand{\bE}{\mathbf{E}}

\newcommand{\cH}{\mathcal{H}}
\newcommand{\R}{\mathbf{R}}
\newcommand{\C}{\mathbf{C}}
\newcommand{\bP}{\mathbf{P}}

\newcommand{\al}{\alpha}
\newcommand{\be}{\beta}
\newcommand{\ga}{\gamma}
\newcommand{\de}{\delta}

\newcommand{\beq}{\begin{equation}}
	\newcommand{\eeq}{\end{equation}}
\newcommand{\lra}{\longrightarrow}

\newcommand{\KL}{\mathrm{KL}}

\begin{document}

\centerline{\Large
\bf The Geometry of Quantum Computing}

\medskip

\centerline{\large{{E. Ercolessi$,^{\! *, \dagger}$ R. Fioresi$,^{\!\star,\dagger}$ 
T. Weber$^{\#,\dagger}$}}
\footnote{This research was supported by Gnsaga-Indam, by
COST Action CaLISTA CA21109, HORIZON-MSCA-2022-SE-01-01 CaLIGOLA,
PNRR MNESYS, PNRR National Center for HPC, Big Data and Quantum Computing. }}

\vskip 0.5 cm
\centerline{$^*${\sl Dipartimento di Fisica e Astronomia,
Universit\`a di Bologna}}

\centerline{{\sl  
Via Irnerio 46 - 40126, Bologna, Italy}}

\centerline{{elisa.ercolessi@unibo.it}}

\bigskip

\centerline{$^\star${\sl FaBiT, Universit\`{a} di
Bologna}}

\centerline{\sl via S. Donato 15, I-40126 Bologna, Italy}

\centerline{{rita.fioresi@unibo.it}}

\bigskip

\centerline{
$^\#${\sl Dipartimento di Matematica, Universit\`{a} di
Bologna}}

\centerline{\sl Piazza di Porta S. Donato 5, I-40126 Bologna, Italy}

\centerline{{thomasmartin.weber@unibo.it}}

\bigskip

\centerline{$^\dagger${\sl Istituto Nazionale di Fisica Nucleare}}

\centerline{\sl Sezione di Bologna,  I-40126 Bologna, Italy}

\begin{abstract}
In this expository paper we present a brief introduction to
the geometrical modeling of some quantum computing problems.
After a brief introduction to establish the terminology, we
focus on quantum information geometry and $ZX$-calculus,
establishing a connection between quantum computing questions and
quantum groups, i.e. Hopf algebras.
\end{abstract}

\section{Introduction}

The idea of using the laws of quantum mechanics for a new approach
to computer algorithms is due to R. Feyman, who introduced the concept
of quantum computing in a seminal talk in 1981 \cite{feynman}.
Feyman himself, a few years later \cite{feynman1987},
analyzed the concept of universal
quantum computer built from the simplest quantum mechanical system
(a two-level system or qubit) and making use of elementary
quantum gates. So Feynman, for the first time, introduced the idea
of a computer, functioning according to {\sl non-binary} i.e.
{\sl quantum} logic.
Feyman also gave a graphical description of the quantum logic gates,
a notation we still use today, as many of Feyman extraordinary ones.
Around the same time, Yu. Manin \cite{manin} suggested a quantum approach
to information theory from a mathematical point of view.
Later on Manin enhanced the theory of quantum groups with a new geometric approach \cite{manin2}.
In 1985 Deutsch \cite{deutsch} formalized the notion of quantum computer
and was the first to raise the question of ``quantum advantage'', which later
on inspired Shor towards the formulation of a quantum algorithm
for factoring large numbers more efficient than the classical ones \cite{shor1, shor2}.

\medskip
We are unable to properly account for the many discoverings in this
fast developing subject, we send the reader to the exhaustive treatments
by Preskill \cite{preskill1, preskill2, preskill3}. 

\medskip
In the present note
we want to show how quantum computing, quantum information geometry and quantum
groups, in the form of Hopf algebras, can be related, providing new models for interesting
quantum computing questions. Moreover, our elementary treatment
of the topics, hopefully will allow experts in different fields,
to share their knowledge towards a more unified and sound
geometric theory of quantum computing.

\medskip
Our paper is organized as follows.

\medskip
In Sec. \ref{overview-sec}, we establish the terminology
introducing the notions of {\sl qubit}, {\sl density operator}
in \ref{qbit-sec},
and {\sl quantum logic gate} in \ref{qlog-sec},
that will be key for our subsequent
treatment.

\medskip
In Sec. \ref{info-sec}, we first introduce basic concepts of
information geometry, as the {\sl Fisher matrix} in
\ref{fisher-sec}, and then we explore
their quantum information counterparts in \ref{qfisher-sec},
leading to the {\sl Quantum
Geometric Tensor} in \ref{qgt-sec},
a natural K\"ahler metric on the space of qubits.

\medskip
Finally in Sec. \ref{zx-sec}, we discuss $ZX$-calculus as an effective
method to describe quantum logic gates and quantum circuits, making
use of the Hopf algebra language and leading to a
surprising connection with the quantum group theory.

\medskip
Our goal is to introduce mathematicians and physicists to some theoretical
questions arising from the quantum computing world. 
We are convinced that geometric and theoretical physical modeling
can greatly help this fast developing field.

\section{An overview on Quantum computing}\label{overview-sec}


In this section we give the basics of quantum
computing theory, establishing a dictionary, which will enable us to
state some interesting mathematical questions 
that we will discuss in the next sections.
Our treatment is very terse, we {refer}
the interested read to \cite{preskill3} (and refs. therein)
and to \cite{as}. 

\subsection{The Qubit}\label{qbit-sec}

In a classical computer, the fundamental unit of information is
the \textit{bit}, which can assume either 0 or 1 value.
In quantum computing, the corresponding notion is the quantum bit or, in short, \textit{qubit}.
A qubit differs from a classical bit since its state is represented by any  
\textit{superposition} of two independent  states that correspond to the classical 0 and 1. Mathematically,
we represent {the state $|\psi\rangle$ of a qubit as the linear combination 
of two basis states $|0\rangle$ and $|1\rangle$, the \textit{computational basis}, 
that we identify with the canonical basis of $\C^2$:
\beq\label{qubit}
|\psi\rangle = \al |0\rangle +\be |1\rangle, \qquad |\al|^2+|\be|^2 =1,
\,\, \al,\be \in \C
\eeq
together with a {\sl phase condition} that we detail below.
Hence $|\psi\rangle$ is an element of the Hilbert space $\cH=\C^2$, {of unit norm. This condition
follows from the probabilistic interpretation} of the coefficients $\al$ and $\be$:
the result of a measurement { in the computational basis} finds the qubit in the state
$|0\rangle$ or in the state $|1\rangle$ with probability $|\al|^2$ or $|\be|^2$. {Notice that this interpretation 
does not distinguish between the states $|\psi\rangle$ and $e^{i \gamma} |\psi\rangle$, for any phase $\gamma$.}

\medskip
We write the qubit $|\psi\rangle$ using the {Dirac} ket notation, as customary in quantum mechanics 
to represent a quantum mechanical state.


Since we pursue a 
geometrical interpretation, we also think
of $|\psi\rangle$ as a {complex} {\sl line} in the vector space $\cH$.
Hence in eq. (\ref{qubit}), 
we need to choose a representative of such line by looking at the
intersections of the line with the sphere (the condition  $|\al|^2+|\be|^2 =1$)
and then identifying {all points that differ just by a phase factor $e^{i \gamma}$.}

\medskip
For this reason, we represent the qubit $|\psi\rangle$ parametrically as
\beq\label{phase-eq}
|\psi\rangle = e^{i\ga} \left[ \cos\left(\frac{\theta}{2}\right)|0\rangle +
 e^{i\varphi} \sin\left(\frac{\theta}{2}\right)|1\rangle\right]  
\eeq
where we can set  the parameter $\ga=0$ (phase) to account for the identification
{described above}. 
So a qubit $|\psi\rangle$ is also equivalently given by
the two parameters $(\varphi, \theta)$, 
or by eq. (\ref{phase-eq}),
together with the phase condition $\ga=0$.

We can then view a
qubit as an element of a 2 dimensional (real) sphere in $\R^3$,
called  the \textit{Bloch sphere}, as in Fig. \ref{bloch}. 
Mathematically this is the
{\it Riemann sphere}, and it is identified with $\bP^1(\C)$
the projective line, consisting of lines in $\C^2$.
Notice that, conventionally, the basis vectors $|0\rangle$ and $|1\rangle$
are respectively the north and south pole of the Bloch sphere.

\begin{figure}[h!]
\begin{center}
\includegraphics[width=3in]{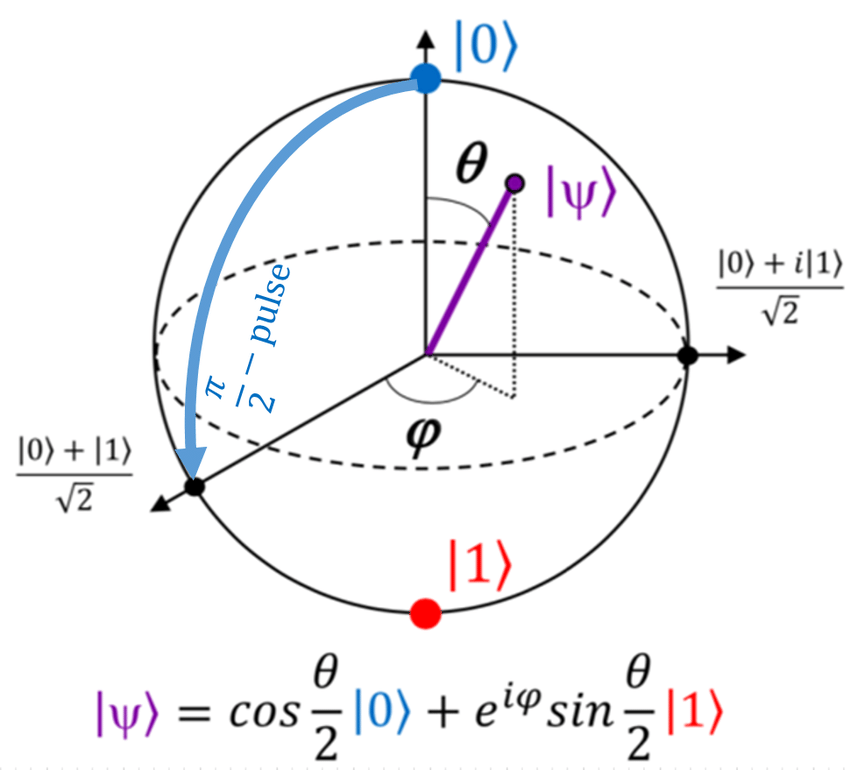}
\caption{The Bloch Sphere}\label{bloch}
\end{center}
\end{figure}

Later on, when we introduce quantum information geometry,
we will equip the Bloch sphere with a natural metric,
the \textit{Fubini-Study metric}, coming from its identification
with the complex projective line (Riemann Sphere).

\medskip
When we measure the qubit, we have that its state $|\psi\rangle$ collapses into
one of the basis states  $|0\rangle$,  $|1\rangle$ with a certain
probability. The problem with quantum computing, in comparison with
classical computing, is that we cannot measure a state with arbitrary accurate precision, since the act of measuring will affect the state
itself. As a consequence of this quantum mechanical phenomenon,
the quantum algorithms need to be suitably adapted and are not
simple generalizations of the classical ones. Moreover, they are
more prone to errors, and the correction of errors leads to
\textit{fault tolerant algorithms} that we mention later on
in our note, see also \cite{kitaev1, kitaev2, dklp}.

\medskip
Besides the above mentioned single or $1$ qubit (\ref{qubit}),
we may have  {a $2$ qubit system, whose states} are elements 
of $\C^2 \otimes \C^2$. 
We represent a $2$ qubit system in  the basis:
$$
|00\rangle, \quad |01\rangle, \quad |10\rangle, \quad |11\rangle
$$
called again the \textit{computational basis}, 
where we use the notation:
\beq\label{basis}
|ij\rangle := |i\rangle \otimes |j\rangle
\eeq


More in general, we can speak of  {an $n$ qubit system} 
in a similar fashion: {states are expressed as linear
superposition (combination) of the computational basis states: 
\beq\label{nbasis}
|00\dots 00\rangle,  \quad |00\dots 01 \rangle\quad |00\dots 10 \rangle, \quad
\dots \quad |11\dots11\rangle
\eeq
Notice that the $2^n$ states of the computational basis are ordered
from $00...0$ to $11\dots1$ according to the binary representation of the
integer $i$, $1=0,1,\dots, 2^n-1$} enumerating them. 
The space
of $n$ qubits is $\cH:=\C^2 \otimes \dots \otimes \C^2$
($n$ tensor product). Hence a state of {$n$ qubits} is an element of the Hilbert space
$\cH \cong \C^N$, for $N=2^n$
$$
|\psi\rangle=\al_0|00\dots 0\rangle+\al_1 |00\dots 1 \rangle + \dots
+\al_{2^n-1}|11\dots1\rangle, \quad \sum_{i=0}^{2^n-1} |\al_i|^2=1
$$
The last condition takes into account the above mentioned fact
about measurements in quantum mechanics: after
a measurement, the state $|\psi\rangle$ is found in the $i$th state of the computational basis 
with probability $|\al_i|^2$.
We should also take into account the phase condition, 
but we do not write such condition explicitly here.

A vector $\psi\in \C^2 \otimes \dots \otimes \C^2$  describing the state of $n$ qubits is said to be 
 \textit{entangled} if it is not an indecomposable tensor, that is if we cannot write
it as $|\psi\rangle=u_1 \otimes \dots \otimes u_n$, with $u_i \in \C^2$.


\medskip
The states considered up to this points are called {\it pure} states, to be distinguished from the case
in which a system of $n$ qubits can be described by a {\it mixed} state, i.e. by a statistical mixture  
defined by an ensemble of states $\{ |\psi_j\rangle\}_{j=1}^m$
that can occur with a probability $p_j$, with 
$0\leq p_j\leq 1$ and $\sum _{j=1}^m p_j=1$.
Clearly, a pure state is a particular example of a mixed
state with $p_1=1, \; p_2=\dots=p_m=0$.

Both pure and mixed states can be represented by means of the so-called
{\it density operator}
\beq\label{rho-mixed}
\rho =\sum _{s}p_{s}|\psi _{s}\rangle \langle \psi _{s}|
\eeq
where $\langle \psi _{s}| \in \cH^*$ is the dual vector of $|\psi _{s}\rangle$.

It is immediate to verify that $\rho$ is a self-adjoint, (semi)positive definite operator with  unit trace.
Also, one can prove 
that it is idempotent (hence it is a projection
operator) if and only if the state is pure. 

Once an ON (orthonormal) basis is fixed, we can express $\rho$ as a matrix
and speak of \textit{density matrix}. With an abuse of notation we shall use
$\rho_\psi$ to refer to both density matrix and operator.

For example, for a $1$ qubit system in the pure state $|\psi\rangle=\al|0\rangle + \be|1\rangle$,  
we have
in the computational basis:
$$
\rho_\psi=\begin{pmatrix} \al \\ \be \end{pmatrix}
\begin{pmatrix} \overline{\al} & \overline{\be} \end{pmatrix}=
\begin{pmatrix} |\al|^2 & \al \overline{\be} \\ \overline{\al}\be & |\be|^2
\end{pmatrix}
$$
while the mixed state defined by the ensemble $\{|0\rangle, |1\rangle\}$ with $p_0=|\alpha|^2, p_1=|\beta|^2$ we have:
$$
\rho_\psi=\begin{pmatrix} \al \\ \be \end{pmatrix}
\begin{pmatrix} \overline{\al} & \overline{\be} \end{pmatrix}=
\begin{pmatrix} |\al|^2 & 0 \\ 0 & |\be|^2
\end{pmatrix}
$$

\subsection{Quantum Logic Gates} \label{qlog-sec} 

In classical computing we have logic gates operating on bits.
For example, the classical logic gate ``NOT'' operates as follows:
\beq\
\begin{array}{c}
0 \mapsto 1\\
1 \mapsto 0
\end{array}
\eeq

{Quantum logic gates operate on
states which are linear combinations of the computational basis elements,
as in (\ref{basis}) for $1$ qubit and in (\ref{nbasis}) for $n$ qubits. In order 
to preserve probabilities, quantum logic gates are represented by unitary matrices $U$, i.e.
$U^\dagger U=I$.}

Let us see examples of logic gates for single qubits; these examples
will be especially relevant for the $ZX$ calculus that we will
treat more in detail in the last section.
We can take advantage of the Pauli matrices $X$, $Y$, $Z$ and the operator 
$H$ called the {\it Hadamard operator}:
$$
X=\begin{pmatrix} 0 & 1 \\ 1 & 0 \end{pmatrix}, \qquad
Y=\begin{pmatrix} 0 & -i \\ i & 0 \end{pmatrix}, \qquad
Z=\begin{pmatrix} 1 & 0 \\ 0 & -1 \end{pmatrix}, \qquad
H=\frac{1}{\sqrt{2}}\begin{pmatrix} 1 & 1 \\ 1 & -1 \end{pmatrix}.
$$
The Hadamard operator transforms the computational basis $|0\rangle$,  $|1\rangle$ 
{(which are the eigenvectors of the $Z$ operator) into the {\it Hadamard 
basis} corresponding to the eigenvectors of the $X$ operator. The Hadamard basis is denoted by}
$$
|+\rangle:=H |0\rangle = \frac{1}{\sqrt{2}}\left(|0\rangle + |1\rangle\right), \qquad
|-\rangle:=H |1\rangle = \frac{1}{\sqrt{2}}\left(|0\rangle - |1\rangle\right)
$$
These operators are also given graphically in Fig. \ref{qgates}.

\begin{figure}[h!]
\begin{center}
\includegraphics[width=2.5in]{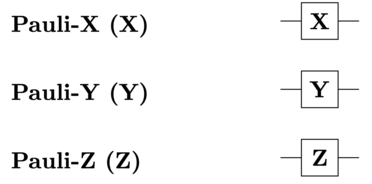}
\caption{Quantum logic gates for $1$ qubit}\label{qgates}
\end{center}
\end{figure}


We now want to discuss an example of a quantum logic gate for a $2$ qubit system: the CNOT,
{\it control not}. The unitary transformation expressing the CNOT is given
in the computational basis by the matrix:
$$
U_{\mathrm{CNOT}}=\begin{pmatrix} 1 & 0 & 0 & 0 \\0 & 1 & 0 & 0 \\0 & 0 & 0 & 1 \\0 & 0 & 1 & 0
\end{pmatrix}
\qquad
\begin{array}{c}
|00\rangle \mapsto |00\rangle \\
|01\rangle \mapsto |01\rangle \\
|10\rangle \mapsto |11\rangle \\
|11\rangle \mapsto |10\rangle 
\end{array}
$$
and is graphically represented as in Fig. \ref{cnot}.

\begin{figure}[h!]
\begin{center}
\includegraphics[width=3in]{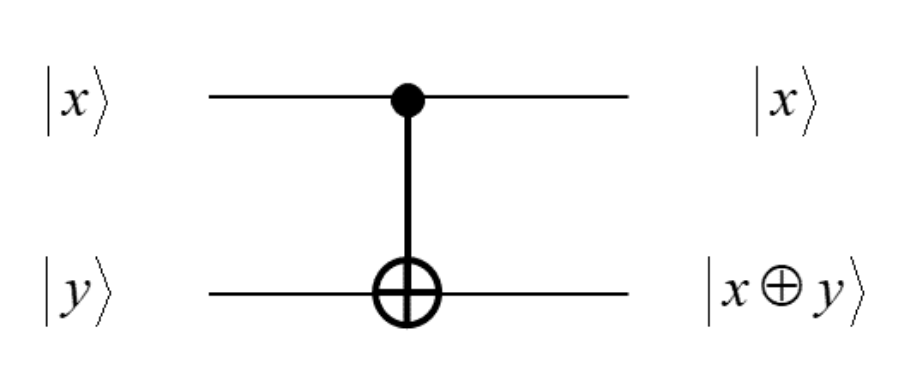}
\caption{CNOT logic gate}\label{cnot}
\end{center}
\end{figure}

This transformation is called ``control not'', since it reverses the computational
basis state in Fig. \ref{cnot} represented by  $|y\rangle$ (i.e. $y=0,1$)
if only if the {\sl control state} 
$|x\rangle=|1\rangle $. 
Notice that the sum $x \oplus y$ means sum modulo 2.

\medskip
Despite the specificity of the quantum gates described above, we have a beautiful
theoretical result, \cite{shor3}, showing that actually single qubit gates and CNOT
is all we need to construct the most general quantum logic gates.

\begin{theorem}
The set of gates that consists of all 1 qubit quantum gates (namely $U(2)$ operators)
and the 2 qubit gate CNOT is universal in the sense that all unitary operations on 
$n$ qubits (namely $U(2^n)$) can be expressed as compositions of these gates.
\end{theorem}

This theorem is very significant, since it shows the universality of the set consisting of unitary 
single qubit operators and the CNOT: from this set we can then get the operators to build any
quantum algorithm. Thus, this
is the counterpart of the classical observation that with classical NAND gates, we can realize all logic functions: 
AND, OR, NOT { 
and therefore any classical algorithm}. 

\medskip

Still, it is not clear how (and if!) a quantum computer can solve any classical computability question;
essentially the obstruction comes from two main reasons.
First, because 
the classical NAND is a non-reversible gate, while quantum gates are necessarily unitary and hence invertible.
Then, we have the
so-called {\it No-cloning theorem}, that forbids the existence of a quantum gate
that allows to copy a generic state of a qubit.

\begin{theorem} {\bf (No-cloning theorem).}
There is no unitary operator $U$ on $\cH \otimes \cH$ such that for all pure
states $|\phi \rangle$ and $|B\rangle$ in $\cH$, we have:
$$
U(|\phi \rangle \otimes |B\rangle)=e^{i\alpha (\phi,B)}|\phi \rangle\otimes |\phi \rangle
$$
for some real number $\alpha$ depending on $\phi$ and $B$.
\end{theorem}
\begin{proof} Assume such $U$ exists, and assume for simplicity $\alpha=0$.
Given the two states $|\phi\rangle$, $|\psi\rangle$
in $\cH$ we write:
$$
U:\cH \otimes \cH \lra \cH \otimes \cH,
\quad U(|\phi \rangle \otimes |B\rangle)=|\phi \rangle \otimes |\phi\rangle, \quad
U(|\psi \rangle \otimes |B\rangle)=|\psi \rangle \otimes |\psi\rangle
$$
Since $U$ preserves the hermitian product, we have that
$$
|\phi \rangle \otimes |B\rangle \cdot |\psi \rangle \otimes |B\rangle=
|\phi \rangle \otimes |\phi\rangle \cdot |\psi \rangle \otimes |\psi\rangle
$$
which gives the equality:
$$
\langle\psi|\phi \rangle \langle B |B\rangle=
\langle\psi|\phi \rangle  \langle \psi|\phi\rangle
\implies \langle\psi|\phi \rangle = \langle\psi|\phi \rangle^2 
$$
yielding either $\langle\psi|\phi \rangle=0$, i.e. $\psi$, $\phi$ orthogonal
or $\langle\psi|\phi \rangle=1$, i.e. $\phi=\psi$. 
\end{proof}

We overcome this difficulty by introducing a suitable 3 qubit gate (e.g. the Toffoli gate or CCNOT,
see Fig. \ref{ccnot}).

\begin{figure}[h!]
\begin{center}
\includegraphics[width=4in]{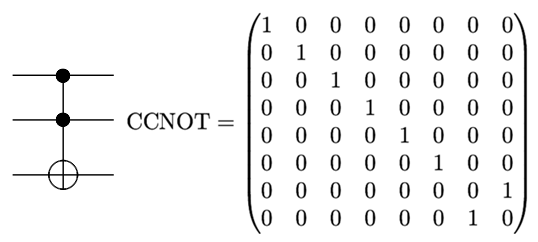}
\caption{CCNOT logic gate}\label{ccnot}
\end{center}
\end{figure}

Such gate allows to realize both a copy of the state of a qubit
and the NAND operation between two qubits, by suitably choosing the entry of the third qubit.


When performing a quantum algorithm on a real physical platform, the interaction of the system with the environment,
due to either measurements or noise, can degrade the information. This differs from what happens in a real classical
device, because a quantum channel that describes either a measurement or a noisy evolution does not simply flip
the state of a qubit in a random way, but can actually transform a pure state into a mixed one, thus resulting
in a loss of information. Also, error correction schemes are more difficult to implement because of the No-cloning theorem.

Still, we can develop fault tolerant algorithms to minimize the damage, and we have an important 
theoretical result, called the {\sl threshold theorem}. This is the analog of the classical
Von Neumann theorem, and states that a quantum computer with a physical error rate below a certain threshold can,
through application of quantum error correction methods, reduce the error rate to arbitrarily low levels.
Hence, we have the hope to create fault-tolerant algorithms
and feasible quantum computing in general. We invite the reader to look
at the works \cite{kitaev1, kitaev2, dklp} for topological approaches to such algorithms.


\section{Classical and Quantum Information Geometry}\label{info-sec}
Information geometry is a new field at the intersection of statistics,
probability and geometry, that
takes advantage of techniques from differential and Riemannian geometry
to study {\sl statistical manifolds}, that is manifolds whose points represent
families of probability distributions.
We focus on the most important notion of classical
Fisher Information Matrix, its quantum counterpart, the
Quantum Fisher Matrix, arriving to define
the Quantum Geometric Tensor, providing a natural hermitian
metric on the space of qubits.

\subsection{The Fisher Information Matrix}\label{fisher-sec}
A central object
in Information Geometry is the Fisher information matrix
(see \cite{fz} for the notation and the refs. therein).

\begin{definition}
Let $p(\theta)$ be a discrete empirical probability distribution, depending
on parameters $\theta \in \R^n$ (or to a statistical manifold embedded in $\R^n$).
  We define $F$ the
  \textit{Fisher information matrix} or \textit{Fisher matrix}
  for short, as
  \beq\label{fisher-def}
F_{ij}(\theta)=\bE_p[\partial_{\theta_i}\log
p(\theta)\partial_{\theta_j}\log p(\theta)]
  \eeq
where $\bE_p$ represents the expected value with respect to
the probability distribution $p$.
\end{definition}

We can also express more concisely the Fisher matrix as:
$$
F(\theta)=\bE_p[\nabla \log p(\theta) (\nabla \log p(\theta))^t]
$$
where we understand $\nabla \log p(\theta) $ as a column vector and $t$
denotes the transpose.

In his pioneering work \cite{amari}, Amari gave a geometric
interpretation of the Fisher matrix, as a meaningful metric
for the statistical manifold and provided a new method
of gradient descent based on his considerations. Later on,
this treatment inspired the quantum counterparts \cite{sikc}.

\medskip
It is therefore important to relate, in the spirit of Amari, 
the Fisher matrix to the information loss
of a probability distribution:
$$
I(\theta)=-\log(p(\theta))
$$
which is also linked to \textit{Shannon entropy} $H(\theta)$ as follows:
$$
H(\theta)=-\bE_p[I(\theta)]=- \sum p_i(\theta) \log(p_i(\theta))
$$
$H$ measures how "diffused" the information in our probability
distribution is; if $p(\theta)$ is a probability mass distribution (i.e. $p_i(\theta)=\de_{ij}$
for a fixed $i$ and $j=1,\dots, n$), then $H(\theta)=0$, while $H(\theta)$ is
maximal if $p_j(\theta)=1/n$.

We now express some well known facts on the Fisher matrix.

\begin{proposition}
  The Fisher matrix 
  is the covariance matrix of the gradient of the information loss.
\end{proposition}

\begin{proof} The gradient of the information loss is 
$$
\nabla_\theta I(\theta)=- \frac{\nabla_\theta p( \theta)}{p( \theta)}
$$
Notice that:
$$
\begin{array}{l}
\bE_{p}(\nabla_\theta I)=
\sum p_i\frac{\nabla_\theta p_i}{p_i}=
\sum_i\nabla_\theta p_i=
\nabla_\theta(\sum_i p_i)=0
\end{array}
$$ 
The covariance
matrix of $\nabla_\theta I(\theta)$ is (by definition): 
$$
\begin{array}{l}
\mathrm{Cov}(I)=\bE_{p}[(\nabla_\theta I -\bE_{p}(\nabla_\theta I))^t
(\nabla_\theta I -\bE_{p}(\nabla_\theta I))]= \\ \\
\qquad=\bE_{p}[(\nabla_\theta I)^t(\nabla_\theta I)]=F(\theta)
\end{array}
$$
\end{proof}

We conclude our brief treatment of Information Geometry by an observation
regarding the metric on the parameter space.

\begin{observation}
We observe that {the Fisher metric is related to the Kullback-Leibler divergence (see \cite{fz}
and refs therein)
that gives the statistical distance between two probability distributions. Indeed, if we consider 
$\KL(p( \theta+\delta w)||p( \theta))$ that measures
how $p( \theta+\delta w)$ and $p( \theta)$ differ, for a small variation of the parameters $\delta w$,
we have:}
$$
\begin{array}{rl}
\KL(p( \theta+\delta w)||p( \theta)) &\cong 
 \frac{1}{2}(\delta w)^t F(\theta) (\delta w) 
+\mathcal{O}(||\delta w||^3)
\end{array}
$$
This clearly shows how the Fisher can be interpreted as a metric
on the statistical manifold.
\end{observation}

\subsection{The Quantum Information Matrix} \label{qfisher-sec}

In quantum metrology, which is
the study of physical measurements via quantum theory, 
the {\sl quantum information matrix} or
{\sl quantum Fisher matrix} is the analog of the Fisher matrix, 
introduced in the previous section. We now give its definition
and main properties; for the details,
we refer the reader to \cite{helstrom}, \cite{liu} and refs. therein.

\medskip
We start with the definition of
{\sl symmetric logarithmic derivative}. This is a generalization of
the usual notion of logarithmic derivative and it is necessary
in the quantum setting, since the scalar expressing
the probability distribution $p$ appearing in
the definition Fisher matrix (see (\ref{fisher-def})) is replaced by
the density operator $\rho$ as in (\ref{rho-mixed}).

\medskip
We assume that $\rho$ is depending on a parameter $\theta$; when $\rho$
is defined for pure states (a special, but important case),
this parameter comes from a parametrization 
{ of the vectors in the Hilbert space.
}

\begin{definition}
We define, implicitly, the \textit{symmetric logarithmic derivative} $L$
of a given density operator $\rho$ as:
\beq\label{log-def}
\partial_\theta \rho =\frac{L \rho + \rho L}{2}
\eeq
\end{definition}
We notice immediately that, whenever $L$ and $\rho$ are scalars,
$L$ is the gradient of the logarithm of $\rho$ as it appears in
(\ref{fisher-def}), that is 
$L=\partial_\theta \rho /\rho=\partial_\theta(\log(\rho))$.

We can derive the explicit expression of $L$, when we fix an
ON (orthonormal) basis $\{|\psi_j\rangle\}$ for qubits. 

\begin{observation}
Let $\{|\psi_j\rangle\}$ be an ON basis for $\cH$. 
{Thanks to the properties of the density matrix $\rho$, we can
consider the ON basis $\{|\psi_j\rangle\}$ of its eigenvectors, with 
$ \rho |\psi_j\rangle = p_j |\psi_j\rangle$, $j=1,\dots,dim(\cH)=N$. 
Here $0\leq p_j\leq 1,\sum p_j =1$ and indeed we can assume that only the 
first $s$ values are different form zero, $s$ being the rank of $\rho$. 

Now, for any  $N \times N$ matrix $A$, 
we can express the $(i,j)$ entry of $A$ as $A_{ij}=\langle\psi_i|A|\psi_j\rangle$. 
Applying this observation to (\ref{log-def}) and using the definition (\ref{rho-mixed}) for
$\rho$, we get:
$$
\begin{array}{rl}
(\partial_\theta \rho)_{ij}&= \langle\psi_i|\partial_\theta \rho |\psi_j\rangle=
\frac{1}{2}\langle\psi_i| (L \rho + \rho L) |\psi_j\rangle=\\ \\
&=\frac{1}{2} \left(\langle\psi_i| L p_j  |\psi_j\rangle +
\langle\psi_i| p_i L  |\psi_j\rangle \right)=\frac{p_i+p_j}{2}L_{ij}
\end{array}
$$
Notice that $p_i+p_j=0$ if both $i,j>s$ so that we can assume 
$$
L_{ij}=0 \; , \;  i,j=s,\dots N
$$ 
while the other matrix elements of $L$ are given by}
\beq\label{log-der} 
L_{ij}=\frac{2 (\partial_\theta \rho)_{ij}}{p_i+p_j}
\eeq
\end{observation} 

Now we can define the quantum information matrix.

\begin{definition}
Let $\rho$ be a density operator as in (\ref{rho-mixed}).
We define \textit{quantum information matrix} as
\beq
F_Q = \mathrm{tr}[\rho L^2]
\eeq
where $L$ is the symmetric logarithmic derivative of $\rho$.
\end{definition}

We notice that $F_Q$ is actually an operator, however we adhere
to the most common terminology {\sl quantum information matrix}.
With an abuse of notation, we write $F_Q$ also for the matrix
representing $F_Q$ in a fixed ON basis.

\medskip
There is a strict analogy between $F_Q$ and the
Fisher information matrix $F$ as in (\ref{fisher-def}).

\begin{observation}

Let the notation be as above and let us fix {the ON basis of eigenvectors 
of $\rho$}. Then, we have:
\beq\label{fq-exp}
F_Q=\sum_{k=1}^s \sum_{l=1}^N  p_k L_{kl} L_{lk}
\eeq
In fact
$$
\begin{array}{rl}
F_Q&=\mathrm{tr}[\rho L^2]=\sum_k \langle\psi_k|\rho L^2 |\psi_k\rangle=
\sum_{k,l} \langle\psi_k|\rho|\psi_l\rangle \langle\psi_l| L^2 |\psi_k\rangle=
\\ \\
&=\sum_k p_k\langle\psi_k| L^2 |\psi_k\rangle=
\sum_{k,l} p_k\langle\psi_k| L |\psi_l\rangle \langle\psi_l| L |\psi_k\rangle 
\end{array}
$$
which is the expression (\ref{fq-exp}). { Notice that the sum over $k$ is in practice restricted 
to $k=1,\dots,N$ since $p_k=0$ if $k>s$. }
As one readily sees, this is the analog of the expression (\ref{fisher-def})
for $F$. 
\end{observation}

We now give another useful and interesting expression
of $F_Q$.

\begin{proposition}
Let the notation be as above. Then
\beq\label{prop-qfisher}
F_Q=\sum_{i =1}^s \frac{1}{p_i}({\partial_\theta p_i})^2+
\sum_{i=1}^s\sum_{j=1}^N \frac{4p_i(p_i-p_j)^2}{(p_i+p_j)^2}
|\langle\psi_i| \partial_\theta \psi_j\rangle|^2
\eeq
where $s$ is the rank of $\rho$.
\end{proposition}

\begin{proof}
We substitute the expression (\ref{log-der}) for $L_{ij}$ in (\ref{fq-exp}):
\beq\label{fq2}
F_Q=  \sum_{i=1}^s \sum_{j=1}^N 4p_i
\frac{(\partial_\theta \rho)_{ij}(\partial_\theta \rho)_{ji}}{(p_i+p_j)^2}
\eeq
From the very definition of $\rho$ we get:
$$
(\partial_\theta \rho)_{ij}=\de_{ij} \partial p_i+(p_j-p_i)
\langle\psi_i| \partial_\theta \psi_j\rangle
$$
where we used that
$$
0= { \partial_{\theta} (\langle\psi_i |\psi_j\rangle) }=\langle \partial_\theta \psi_i |
\psi_j \rangle+\langle\psi_i | \partial_\theta \psi_j \rangle 
$$
Notice that if $i \neq j$, $(\partial_\theta \rho)_{ij}=
-(\partial_\theta \rho)_{ji}$, so that:
$$
F_Q=\sum_{i =1}^s \frac{1}{p_i}({\partial_\theta p_i})^2+
\sum_{i=1}^s \sum_{i\neq j, j=1}^N
\frac{|(\partial_\theta \rho)_{ij}|^2}{(p_i+p_j)^2}
$$
Now, substituting the expression for $(\partial_\theta \rho)_{ij}$
we obtain the result.
\end{proof}

The next observation compares, once again, the classical Fisher 
information matrix with the quantum $F_Q$.

{
\begin{observation}
We notice that the first term of the Quantum Fisher (\ref{prop-qfisher}) coincides with the Classical 
Fisher Matrix (\ref{fisher-def}): indeed it is given by changes in the eigenvalues only, i.e. of the probabilities 
$p_k$,  of the density matrix. Instead, the second term of eq. (\ref{prop-qfisher}) arises from changes
in the eigenvectors and it is therefore quantum in nature.
\end{observation}}

\subsection{Quantum Geometric Tensor} \label{qgt-sec}

The quantum geometric tensor is a key object in quantum information
geometry and it is strictly related to the quantum information matrix
introduced in the previous section.

As we remarked in Sec. \ref{qbit-sec}, the space of qubits is identified
with the projective space
$\bP^N(\C)$, that is the space of {\sl rays} or lines in $\cH=
\C^2 \otimes \dots \otimes \C^2$, $N=2^n$.
The projective space is a K\"ahler manifold, hence it has a natural
K\"ahler metric, called the \textit{Fubini-Study metric}, which is a 2-tensor on
the tangent space of $\bP^N(\C)$. It is a well known Hermitian metric; if we fix an holomorphic
frame, its expression is given by:
\beq\label{fs}
ds^{2}:=g_{i{\bar {j}}}\,dz^{i}\,d{\bar {z}}^{j}=
{\frac {(1+z_{i}{\bar {z}}^{i})\,dz_{j}\,d{\bar {z}}^{j}-{\bar {z}}^{j}z_{i}\,
dz_{j}\,d{\bar {z}}^{i}}{\left(1+z_{i}{\bar {z}}^{i}\right)^{2}}}
\eeq
where, as customary, we sum over repeated indices.

The Hermitian matrix of the Fubini–Study metric $(g_{i{\bar {j}}})$ in this frame is
explicitly given by:
$$
{\bigl (}g_{i{\bar {j}}}{\bigr )}={\frac {1}{\left(1+|\mathbf {z} |{\vphantom {l}}^{2}\right)^{2}}}\left({\begin{array}{cccc}1+|\mathbf {z} |^{2}-|z_{1}|^{2}&-{\bar {z}}_{1}z_{2}&\cdots &-{\bar {z}}_{1}z_{n}\\-{\bar {z}}_{2}z_{1}&1+|\mathbf {z} |^{2}-|z_{2}|^{2}&\cdots &-{\bar {z}}_{2}z_{n}\\\vdots &\vdots &\ddots &\vdots \\-{\bar {z}}_{n}z_{1}&-{\bar {z}}_{n}z_{2}&\cdots &1+|\mathbf {z} |^{2}-|z_{n}|^{2}\end{array}}\right)
$$
with ${|\mathbf {z}|}= |z_1|^2 +|z_2|^2 + \dots + |z^n|^2$.
The metric expressed by $(g_{i{\bar {j}}})$ is clearly invariant under unitary transformations.

Furthermore, we notice
that this metric can be derived from a {\it K\"ahler potential} $K$, that is, we can write it as:
$$
g_{i{\bar {j}}}= \frac{\partial ^{2}} {\partial z^{i} \, \partial \bar z^{j}} K, \qquad
K=\log(1+z_{i}{\bar {z}}^{i})
$$

We briefly recall some
well known facts regarding K\"ahler manifolds.
First, recall that a real manifold $X$ is K\"ahler if it is a symplectic manifold $(X, \omega)$
with an integrable almost-complex structure $J$ compatible with the symplectic form $\omega$,
that is:
$$
g(u,v):=\omega (u,Jv)
$$
is a Riemannian metric on $X$ \cite{huy}. Equivalently,
a complex manifold $X$ is K\"ahler if it has 
an hermitian metric $h$ and a closed 2-form $\omega$ given by:
$$
\omega (u,v):=\operatorname {Im} h(u,v) \, (=\operatorname {Re} h(iu,v))
$$


\medskip
On a K\"ahler manifold we then have the following structures,
compatible with each other. 
\begin{itemize}
\item A complex structure $J$, i.e. $J^2=-I$ on the tangent bundle.
\item A symplectic (real) structure $\omega$, i.e. a non degenerate closed 2-form.
\item A Riemannian metric $g$, ($g(X,Y)=\omega(X,JY)$).
\item An hermitian metric $h$ such that (real and imaginary parts):
$$
h=g-i\omega
$$
\end{itemize}

We are ready for the key definition of Quantum Geometric Tensor.

\begin{definition}
We define the \textit{Quantum Geometric Tensor} (QGT) on 
$\bP^N(\C)$ as the (K\"ahler) metric: 
\begin{equation} \label{khaler}
ds^2 = \frac{\langle d \psi \vert d \psi \rangle} {\langle \psi \vert \psi \rangle} -
\frac {\langle d \psi \vert \psi \rangle \; \langle \psi \vert d \psi \rangle}
{\langle \psi \vert \psi \rangle^2 }
\end{equation}
where here { $| \psi \rangle$} is not assumed to have length 1.
\end{definition}

We have the following result, which is just a calculation that can be found in \cite{hou}, App. A.

\begin{proposition}
The QGT is equivalent to the Fubini-Study hermitian metric.
\end{proposition}

We now show that the QGT is strictly related with $F_Q$ the quantum
Fisher (see \cite{fkm} for more details).

\begin{proposition}
The real part of the Quantum Geometric Tensor is equivalent to the metric defined by the quantum Fisher $F_Q$
symmetric operator on rays in $\cH$. 
\end{proposition}

\begin{proof}
{For pure states $\rho=|\psi\rangle\langle\psi|$, it is easy to check that $L=2d\rho = 2(  |d\psi\rangle\langle\psi|+|\psi\rangle\langle d\psi|)$. From a straightforward calculation it follows that
$F_Q = 4 Tr [ \rho (d\rho)^2] =  4 \left(   |d\psi\rangle\langle\psi|+|\psi\rangle\langle d\psi| \right)$, which reproduces eq. (\ref{khaler}) when the vector has unit norm. }
\end{proof}

\section{\texorpdfstring{$ZX$}{ZX}-calculus and Hopf algebras} \label{zx-sec}

		In Sec. \ref{qlog-sec} we introduced the quantum logic gates.
		They are used for \textit{quantum circuits}, which consist
		in a sequence of quantum logic gates and measurements, though
		in the present treatment we shall ignore
		the latter. 
		Fig \ref{fig:circ} represents a quantum circuit.
		
		\begin{figure}[h!]
			\begin{center}
				\includegraphics[width=0.95\textwidth]{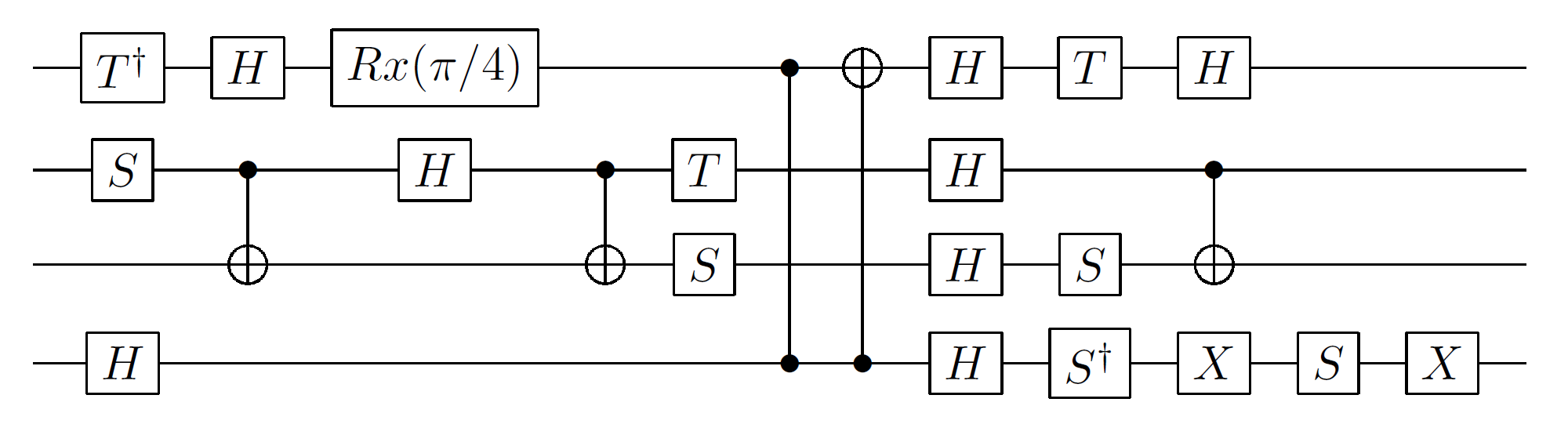}
				\caption{$X$, $Y$, $Z$ are the Pauli matrices gates, $H$ the Hadamard
					gate. The $S$ gate is $e^{i\pi/2}Z$, the $T$ gate is $e^{i\pi/4}Z$,
					$\dagger$ as usual indicates the adjoint.}
				\label{fig:circ}
			\end{center}
		\end{figure}
		
		In this section we want to show how to represent quantum
		circuits via $ZX$-diagrams, and how $ZX$-calculus,
		originally developed in \cite{cd1} \cite{CoeckeDuncan}, can be effectively used to understand and
		optimize simple quantum circuits. Our main references are \cite{MajidZX}
		and the expository treatment \cite{zxcalc} (see also the complete
		list of refs. in \cite{zxcalc}).
		We shall also explore the intriguing connection between the language of
		$ZX$-calculus and quantum groups, as expressed in  \cite{MajidZX}.
		
		\subsection{Quantum circuits and \texorpdfstring{$ZX$}{ZX} diagrams}
		We first establish the conventions that allow us to translate
		a quantum circuit into a \textit{$ZX$-diagram}. The main components of a $ZX$-diagram
		are green (or $Z$) and red (or $X$)
		\textit{spiders} operating on qubits as unitary operators as follows: 
		\begin{equation}\label{spiders}
			\begin{ZX}[row sep=5mm]
				\zxN{} \arrow[drr,bend right] & \zxN{} \arrow[dr,bend right] & \ldots & \zxN{} \arrow[dl,bend left] & \zxN{} \arrow[dll,bend left] \\
				& & \zxZ{\alpha} \arrow[dll, bend right] \arrow[dl, bend right] \arrow[drr, bend left] \arrow[dr, bend left] & & \\
				\zxN{} & \zxN{} & \ldots & \zxN{} & \zxN{}
			\end{ZX}\quad\colon \cH_n\to \cH_m\quad,\quad
			\begin{cases}
				|0\ldots 0\rangle & \mapsto |0\ldots 0\rangle\\
				|1\ldots 1\rangle & \mapsto e^{\mathrm{i}\alpha}|1\ldots 1\rangle\\
				\text{others} & \mapsto 0
			\end{cases}
		\end{equation}
		\begin{equation}\label{spiders2}
			\begin{ZX}[row sep=5mm]
				\zxN{} \arrow[drr,bend right] & \zxN{} \arrow[dr,bend right] & \ldots & \zxN{} \arrow[dl,bend left] & \zxN{} \arrow[dll,bend left] \\
				& & \zxX{\alpha} \arrow[dll, bend right] \arrow[dl, bend right] \arrow[drr, bend left] \arrow[dr, bend left] & & \\
				\zxN{} & \zxN{} & \ldots & \zxN{} & \zxN{}
			\end{ZX}\quad\colon \cH_n\to \cH_m\quad,\quad
			\begin{cases}
				|+\ldots +\rangle & \mapsto |+\ldots +\rangle\\
				|-\ldots -\rangle & \mapsto e^{\mathrm{i}\alpha}|-\ldots -\rangle\\
				\text{others} & \mapsto 0
			\end{cases}
		\end{equation}
		where $\cH_n=\C^2\otimes \dots \otimes \C^2$ ($n$ times) and also $n=0$ or $m=0$ are allowed.
		If $\alpha=0$ we omit the symbol of $\alpha$.

		So, for example, we can establish the correspondence
		between elementary operations on 1 qubits and $ZX$-diagrams as in Fig. 
  \ref{fig:zx-elem}. We will detail the
 rules involved in the identifications in the next section. Notice
 that customarily the quantum circuits are read from left to right (see  Fig.
 \ref{fig:circ}), while the $ZX$-diagrams {are depicted} from top to bottom. We will follow
 these conventions.
		
            \begin{figure}
            \begin{center}
                \begin{tabular}{rll}
                \begin{ZX}[row sep=5mm]
                \zxZ{} \arrow[d]\\
                \zxN{}
                \end{ZX}
                & $=|0\rangle+|1\rangle$ & $=\sqrt{2}|+\rangle$ \\
                & & \\
                \begin{ZX}[row sep=5mm]
                \zxX{} \arrow[d]\\
                \zxN{}
                \end{ZX}
                & $=|+\rangle+|-\rangle$ & $=\sqrt{2}|0\rangle$ \\
                     & & \\
                \begin{ZX}[row sep=5mm]
                \zxN{} \arrow[d]\\
                \zxZ{}
                \end{ZX}
                & $=\langle 0|+\langle 1|$ & $=\frac{1}{\sqrt{2}}\langle +|$ \\
                & & \\
                \begin{ZX}[row sep=5mm]
                \zxN{} \arrow[d]\\
                \zxX{}
                \end{ZX}
                & $=\langle +|+\langle -|$ & $=\frac{1}{\sqrt{2}}\langle 0|$ \\
                     & & \\
                \begin{ZX}[row sep=5mm]
                \zxN{} \arrow[d]\\
                \zxZ{\alpha} \arrow[d]\\
                \zxN{}
                \end{ZX}
                & $=|0\rangle\langle 0|+e^{i\alpha}|1\rangle\langle 1|$ & $=:Z_\alpha$ \\
                     & & \\
                \begin{ZX}[row sep=5mm]
                \zxN{} \arrow[d]\\
                \zxX{\alpha} \arrow[d]\\
                \zxN{}
                \end{ZX}
                & $=|+\rangle\langle +|+e^{i\alpha}|-\rangle\langle -|$ & $=:X_\alpha$
                \end{tabular}
                \caption{Examples of elementary $1$ qubit operations}
                \label{fig:zx-elem}
                \end{center}
            \end{figure}
		
		We further consider the \textit{Hadamard gate}
		\begin{equation*}
			\begin{ZX}[row sep=5mm]
				\zxN{} \arrow[d] \\
				\zxH{} \arrow[d] \\
				\zxN{}
			\end{ZX}\quad=\quad
			\frac{1}{\sqrt{2}}\begin{pmatrix}
				1 & 1 \\
				1 & -1
			\end{pmatrix}
		\end{equation*}
  which we represent as a yellow square in the circuit and the \textit{normalization constant} $\lozenge=\sqrt{2}$.
		Given a quantum circuit, we can then immediately draw the corresponding $ZX$-diagram.


		
		Vice-versa, given a $ZX$-diagram, one can write the quantum circuit
		associated with it (see \cite{bkw} for a full discussion of such reconstruction).
  This means that the rules of the $ZX$-calculus are complete and that we can use them to study fundamental problems, such as quantum circuit optimization and or quantum error correction codes.
		
		\medskip
		There are some mathematically obvious rules regarding quantum circuits,
		that can be translated into operations on the spiders that model them. This originates
		the $ZX$-calculus, that we describe below.
		Such calculus allows to simplify a $ZX$-diagram and then translate such simpler version to a
		quantum circuit {\sl equivalent} to the starting one. This allows for a more
		effective and efficient realization of simple quantum circuits.
		
		\subsection{The \texorpdfstring{$ZX$}{ZX}-calculus}
				
				The $ZX$-calculus is a graphical language to describe linear maps $\cH_n\to \cH_m$ on qubits. 
				For example, the identity $\mathrm{id}_\cH \colon \cH\to \cH$, ($\cH=\C^2$),
				is written as a single string, while the identity $\mathrm{id}_{\cH_2}\colon \cH_2\to \cH_2$ on $\cH_2=\cH\otimes \cH$ corresponds to two strings
				\begin{equation*}
					\mathrm{id}_\cH=\begin{pmatrix}
						1 & 0 \\
						0 & 1
					\end{pmatrix}
					=\quad\begin{ZX}[row sep=5mm]
						\zxN{} \arrow[ddd]\\
						\\
						\\
						\zxN{}
					\end{ZX}\qquad,\qquad
					\mathrm{id}_{\cH^2}=\begin{pmatrix}
						1 & 0 & 0 & 0 \\
						0 & 1 & 0 & 0 \\
						0 & 0 & 1 & 0 \\
						0 & 0 & 0 & 1
					\end{pmatrix}
					=\quad\begin{ZX}[row sep=5mm]
						\zxN{} \arrow[ddd]\\
						\\
						\\
						\zxN{}
					\end{ZX}\quad\begin{ZX}[row sep=5mm]
						\zxN{} \arrow[ddd]\\
						\\
						\\
						\zxN{}
					\end{ZX}\qquad,
				\end{equation*}
				flowing from top to bottom. Some non-trivial operations are given by the ``spiders" with an angle $\alpha\in\mathbb{R}$.
				
				We now list some of such operation, forming the backbone of $ZX$-calculus rules. For this we follow \cite[Section 2.2]{CoeckeDuncan}.

\subsubsection*{The \textbf{T}-rule, or, ``only topology matters"}

This rule postulates that we can straighten, stretch, bend or twist wires without changing the linear operation represented by the $ZX$-diagram. For example
\begin{equation*}
\begin{ZX}[row sep=5mm]
\zxN{} \arrow[d] & [5mm] & [5mm] \\
\zxN{} \arrow[dr,bend right] & & \\
& \zxN{} \arrow[dr,bend left] & \\
& & \zxN{} \arrow[d] \\
& & \zxN{}
\end{ZX}\quad=\quad
\begin{ZX}[row sep=5mm]
\zxN{} \arrow[d] \\
\zxN{} \arrow[d] \\
\zxN{} \arrow[d] \\
\zxN{} \arrow[d] \\
\zxN{}
\end{ZX}\qquad\text{ or }\qquad
\begin{ZX}[row sep=5mm]
\zxN{} \arrow[d] & & & & \zxN{} \arrow[d] \\
\zxN{} \arrow[drr,bend right] & & & & \zxN{} \arrow[dll,bend left] \\
& & \zxN{} \arrow[drr,bend left] \arrow[dll,bend right] & & \\
\zxN{} \arrow[d] & & & & \zxN{} \arrow[d] \\
\zxN{} & & & & \zxN{}
\end{ZX}\quad=\quad
\begin{ZX}[row sep=5mm]
\zxN{} \arrow[d] & & & \zxN{} \arrow[d] \\
\zxN{} \arrow[d] & & & \zxN{} \arrow[d] \\
\zxN{} \arrow[d] & & & \zxN{} \arrow[d] \\
\zxN{} \arrow[d] & & & \zxN{} \arrow[d] \\
\zxN{} & & & \zxN{}
\end{ZX}
\end{equation*}
It is crucial that the wires are not detached from their entry or exit points, while their position can be changed, also in relation to other wires. This also applies to points of connection with green and red vertices or Hadamard gates. In particular, the number of incoming and outgoing wires are invariants under the \textbf{T}-rule.

\subsubsection*{The \textbf{S}-rules}

In a nutshell, the \textbf{S}-rules entail that all operations involving only red or green vertices correspond to red or green spiders, where only the incoming and outgoing wires matter and we sum the phases.

\begin{center}
	\includegraphics{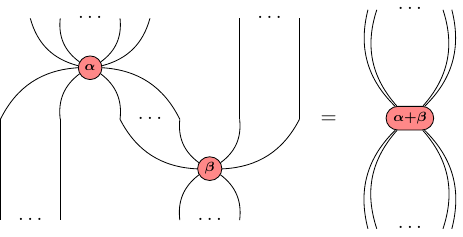}
\end{center}
\begin{center}
	\includegraphics{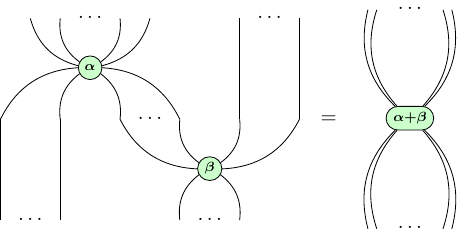}
\end{center}

As a consequence of the \textbf{T} and \textbf{S} rules it follows that $\mathcal{H}$ inherits two algebra and two coalgebra structures.

\begin{corollary}\label{prop-un}
The $1$-qubit space $\mathcal{H}$ is an algebra, where

\medskip\noindent
\begin{center}
        $m_{\textcolor{red}\bullet}=$\begin{ZX}[row sep=5mm]
	\zxN{} \arrow[dr,bend right] & [5mm] & [5mm]
	\zxN{} \arrow[dl,bend left]\\
	& \zxX{} \arrow[d] &\\
	& \zxN{} &
        \end{ZX}\qquad
      and\qquad 
      $m_{\textcolor{green}\bullet}=$\begin{ZX}[row sep=5mm]
        \zxN{} \arrow[dr,bend right] & [5mm] & [5mm]
	\zxN{} \arrow[dl,bend left]\\
	& \zxZ{} \arrow[d] &\\
	& \zxN{} &
	\end{ZX}
\end{center}
 
\medskip\noindent
are associative multiplication operations with unit operations 

\medskip\noindent
\begin{center}
$\eta_{\textcolor{red}\bullet}=$\begin{ZX}[row sep=5mm]
\zxX{} \arrow[d]\\
\zxN{}
\end{ZX}\qquad
and\qquad
$\eta_{\textcolor{green}\bullet}=$\begin{ZX}[row sep=5mm]
\zxZ{} \arrow[d]\\
\zxN{}
\end{ZX},
\end{center}

\medskip\noindent
respectively.

\medskip\noindent
Moreover, $\mathcal{H}$ is a coalgebra, where

\medskip\noindent
\begin{center}
$\Delta_{\textcolor{green}{\bullet}}:=\begin{ZX}[row sep=5mm]
& [5mm] \zxN{} \arrow[d,] & [5mm]\\
& \zxZ{} \arrow[dl,bend right] \arrow[dr,bend left] & \\
\zxN{} & & \zxN{}
\end{ZX}$\qquad and\qquad
$\Delta_{\textcolor{red}{\bullet}}:=\begin{ZX}[row sep=5mm]
& [5mm] \zxN{} \arrow[d,] & [5mm]\\
& \zxX{} \arrow[dl,bend right] \arrow[dr,bend left] & \\
\zxN{} & & \zxN{}
\end{ZX}$
\end{center}

\medskip\noindent
are coassociative comultiplication operations with counit operations

\begin{center}
$\varepsilon_{\textcolor{green}{\bullet}}:=
\begin{ZX}[row sep=5mm]
\zxN{} \arrow[d]\\
\zxZ{}
\end{ZX}$\qquad and\qquad $\varepsilon_{\textcolor{red}{\bullet}}:=
\begin{ZX}[row sep=5mm]
\zxN{} \arrow[d]\\
\zxX{}
\end{ZX}$, 
\end{center}

\medskip\noindent
respectively. 
\end{corollary}

The corresponding (co)algebra axioms, detailed in Appendix \ref{apendixA}, are in fact special cases of the \textbf{S}-rules. 


\subsubsection*{The \textbf{B}-rules}

As a consequence of the previous \textbf{S}-rules we have constructed an associative unital algebra structure $(m_{\textcolor{red}\bullet},\eta_{\textcolor{red}\bullet})$ and a coassociative counital coalgebra structure $(\Delta_{\textcolor{green}\bullet},\varepsilon_{\textcolor{green}\bullet})$ {on $\mathcal{H}$} (see Cor. \ref{prop-un}).
The \textbf{B}-rules will relate such structures with each other to realize an
unnormalized { Hopf algebra} (see {Appendix} \ref{apendixA}). Let us first state the rules, that are
manifest, once the significance of diagrams is translated into quantum circuits.
\begin{center}
	\includegraphics{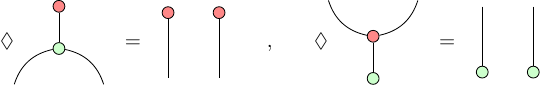}
\end{center}
\begin{center}
	\includegraphics{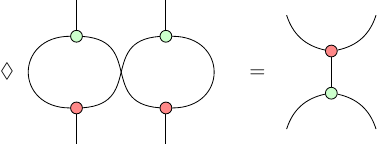}
\end{center}
\begin{center}
	\includegraphics{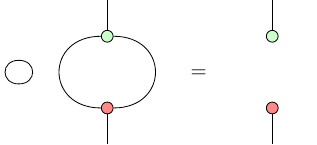}
\end{center}

\begin{proposition}[{\cite[Sect. 2.2.3]{CoeckeDuncan}}]\label{PropHopfAlg}
{Let the notation be as above. Then $(\mathcal{H},m_{\textcolor{red}\bullet},\eta_{\textcolor{red}\bullet},\Delta_{\textcolor{green}\bullet},\varepsilon_{\textcolor{green}\bullet})$ is an unnormalized Hopf algebra with trivial antipode.}
\end{proposition}

\begin{proof}
{
The first diagram above states that $\Delta_{\textcolor{green}\bullet}$ is unital with respect to $\eta_{\textcolor{red}\bullet}$ up to the normalization constant $\lozenge$, while the second diagram states that $m_{\textcolor{red}\bullet}$ is counital with respect to $\varepsilon_{\textcolor{green}\bullet}$ up to the same normalization constant $\lozenge$. The third diagram states that $\Delta_{\textcolor{green}\bullet}$ respects the product $m_{\textcolor{red}\bullet}$ (up to the constant $\lozenge$). In this situation the antipode axiom \eqref{antipode} with respect to the trivial antipode is automatically satisfied, yielding an unnormalized Hopf algebra.}
\end{proof}
Thus, the previous proposition gives an unexpected connection between quantum computing and Hopf algebra (quantum group) theory. We continue to describe the algebraic structure in more detail.
One can show that also $\Delta_{\textcolor{red}\bullet}$ respects the product $m_{\textcolor{green}\bullet}$ (up to the same normalization constant $\lozenge$), which implies that also $(m_{\textcolor{green}\bullet},\eta_{\textcolor{green}\bullet},\Delta_{\textcolor{red}\bullet},\varepsilon_{\textcolor{red}\bullet})$ is an unnormalized Hopf algebra. 
{A natural question is whether $(m_{\textcolor{green}\bullet},\eta_{\textcolor{green}\bullet},\Delta_{\textcolor{green}\bullet},\varepsilon_{\textcolor{green}\bullet})$ or $(m_{\textcolor{red}\bullet},\eta_{\textcolor{red}\bullet},\Delta_{\textcolor{red}\bullet},\varepsilon_{\textcolor{red}\bullet})$ are unnormalized Hopf algebras. This is not the case in general (in fact, only in trivial situations). Instead, the green (or red) algebra and coalgebra together form an \textit{$F$-algebra}, with compatibility condition discussed in Definition \ref{defFAlg}. Together with Proposition \ref{PropHopfAlg}, this is the data of an \textit{F}-Hopf algebra, see Appendix \ref{appFHopf}.}


\begin{proposition}[{\cite{cd1}}]
The pair
					$$
					(A,m_{\textcolor{red}{\bullet}},\eta_{\textcolor{red}{\bullet}},\Delta_{\textcolor{red}{\bullet}},\varepsilon_{\textcolor{red}{\bullet}})\qquad,\qquad
					(A,m_{\textcolor{green}{\bullet}},\eta_{\textcolor{green}{\bullet}},\Delta_{\textcolor{green}{\bullet}},\varepsilon_{\textcolor{green}{\bullet}})
					$$
					of $F$-algebras forms an
    $F$-Hopf algebra.
\end{proposition}

\subsubsection*{The \textbf{K}-rules}

The vertices with phase $\pi$ are "classical", in the sense that they can be copied according to the following rules
\begin{center}
	\includegraphics{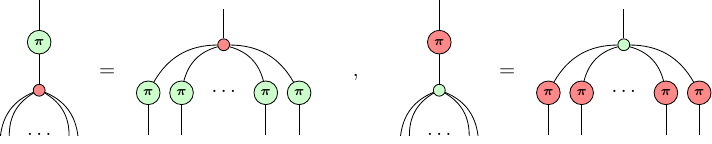}
\end{center}
Furthermore, the vertices of phase $\pi$ can be used to "invert" the phase of the other color, namely
\begin{center}
	\includegraphics{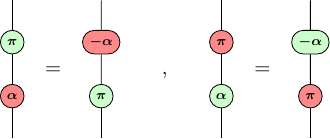}
\end{center}
These are known as the \textbf{K}-rules, where "K" stands for classical (\textit{klassisch} in German).

\subsubsection*{The \textbf{C}-rules}

The \textbf{C}-rules regard \textit{color}.
We can use the Hadamard gate to transform a green vertex into a red one. Explicitly,
\begin{center}
	\includegraphics{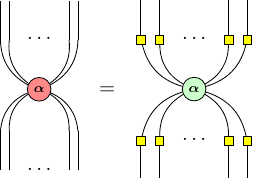}
\end{center}
In particular, the composition of two Hadamard gates yields the identity operation.

\subsubsection*{The \textbf{D}-rules}

Finally, the \textbf{D}-rules determine the normalization constant as the outcome of the following closed circuits
\begin{center}
	\includegraphics{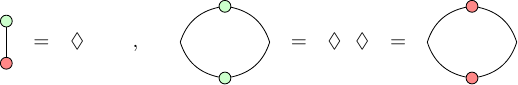}
\end{center}

\bigskip

We show how to apply the rules of the ZX calculus in order to simplify quantum circuits via the two examples in the figures \ref{fig:ZX-comp1} and \ref{ex2} below. Many more can be found in \cite{cd1,CoeckeDuncan}.

        \begin{figure}
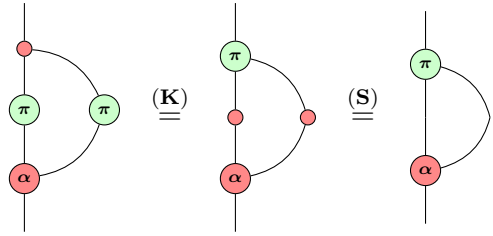

        \begin{center}
        \begin{ZX}[row sep=5mm]
        \zxN{} \arrow[d] & [5mm] \\
        \zxX{} \arrow[d] \arrow[dr,bend left] & \\
        \zxZ{\pi} \arrow[d] & \zxZ{\pi} \arrow[dl,bend left]\\
        \zxX{\alpha} \arrow[d] & \\
        \zxN{} &
        \end{ZX}\quad$\overset{(\textbf{K})}{=}$\quad
        \begin{ZX}[row sep=5mm]
        \zxN{} \arrow[d] & [5mm] \\
        \zxZ{\pi} \arrow[d] \arrow[dr,bend left] & \\
        \zxX{} \arrow[d] & \zxX{} \arrow[dl,bend left]\\
        \zxX{\alpha} \arrow[d] & \\
        \zxN{} &
        \end{ZX}\quad$\overset{(\textbf{S})}{=}$\quad
        \begin{ZX}[row sep=5mm]
        \zxN{} \arrow[d] & [5mm] \\
        \zxZ{\pi} \arrow[d] \arrow[dr,bend left] & \\
        \zxN{} \arrow[d] & \zxN{} \arrow[dl,bend left]\\
        \zxX{\alpha} \arrow[d] & \\
        \zxN{} &
        \end{ZX}
        \caption{Simplifying a ZX diagram}
        \label{fig:ZX-comp1}
        \end{center}
        \end{figure}

        \begin{figure}
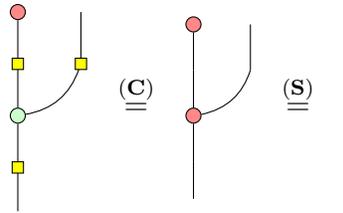

        \begin{center}
        \begin{ZX}[row sep=5mm]
        \zxX{} \arrow[d] & [5mm] \zxN{} \arrow[d] \\
        \zxH{} \arrow[d] & \zxH{} \arrow[dl,bend left]  \\
        \zxZ{} \arrow[d] & \\
        \zxH{} \arrow[d] & \\
        \zxN{} &
        \end{ZX}\quad$\overset{(\textbf{C})}{=}$\quad
        \begin{ZX}[row sep=5mm]
        \zxX{} \arrow[d] & [5mm] \zxN{} \arrow[d] \\
        \zxN{} \arrow[d] & \zxN{} \arrow[dl,bend left]  \\
        \zxX{} \arrow[d] & \\
        \zxN{} \arrow[d] & \\
        \zxN{} &
        \end{ZX}\quad$\overset{(\textbf{S})}{=}$\quad
        \begin{ZX}[row sep=5mm]
        \zxN{} \arrow[d] \\
        \zxN{} \arrow[d] \\
        \zxN{} \arrow[d] \\
        \zxN{} \arrow[d] \\
        \zxN{}
        \end{ZX}
        \caption{Simplifying another ZX diagram}
        \label{ex2}
        \end{center}
        \end{figure}

\appendix

\section{Unnormalized and \texorpdfstring{$F$}{F}-Hopf Algebras }
				\bigskip
				In this appendix, we recall some formal notion regarding unnormalised Hopf algebras
    and $F$-algebras. For more details on Hopf algebras, see \cite{montgomery},
    \cite{kassel}, \cite{majid-found}.
				
				\subsection{Unnormalized Hopf Algebras}\label{apendixA}
				
				We recall the notion of an unnormalized Hopf algebra over a field $\Bbbk$. The axioms are presented in a graphical language, with string diagrams to be read from top to bottom. 

				An \textit{unnormalized Hopf algebra} is a $\Bbbk$-vector space $A$, together with
				\begin{enumerate}
					\item[i.)] a multiplication and unit
					\begin{equation*}
						m_{\textcolor{red}{\bullet}}:=\begin{ZX}[row sep=5mm]
							\zxN{} \arrow[dr,bend right] & [5mm] & [5mm]
							\zxN{} \arrow[dl,bend left]\\
							& \zxX{} \arrow[d] &\\
							& \zxN{} &
						\end{ZX}\quad\colon A\otimes A\to A
						\qquad\qquad
						\eta_{\textcolor{red}{\bullet}}:=  \begin{ZX}[row sep=5mm]
							\zxX{} \arrow[d]\\
							\zxN{}
						\end{ZX}\quad\colon\Bbbk\to A
					\end{equation*}
					satisfying
					\begin{equation*}
						\begin{ZX}[row sep=5mm]
							\zxN{} \arrow[dr,bend right] & [5mm] & [5mm]
							\zxN{} \arrow[dl,bend left] & [5mm] \zxN{} \arrow[dd] \\
							& \zxX{} \arrow[d] & & \\
							& \zxN{} \arrow[dr,bend right] & & \zxN{} \arrow[dl,bend left]\\
							& & \zxX{} \arrow[d] &\\
							& & \zxN{} &
						\end{ZX}\quad=\quad
						\begin{ZX}[row sep=5mm]
							\zxN{} \arrow[dd] & [5mm] \zxN{} \arrow[dr,bend right] & [5mm] & [5mm]
							\zxN{} \arrow[dl,bend left]\\
							& & \zxX{} \arrow[d] &\\
							\zxN{} \arrow[dr,bend right] & & \zxN{} \arrow[dl,bend left] &\\
							& \zxX{} \arrow[d] & & \\
							& \zxN{} & & 
						\end{ZX}
						\qquad\qquad\text{Associativity}
					\end{equation*}
					and
					\begin{equation*}
						\begin{ZX}[row sep=5mm]
							\zxX{} \arrow[d] & [5mm] & [5mm] \zxN{} \arrow[d]\\
							\zxN{} \arrow[dr,bend right] & & 
							\zxN{} \arrow[dl,bend left]\\
							& \zxX{} \arrow[d] &\\
							& \zxN{} &
						\end{ZX}\quad=\quad
						\begin{ZX}[row sep=5mm]
							\zxN{} \arrow[ddd]\\
							\\
							\\
							\zxN{}
						\end{ZX}
						\quad=\quad\begin{ZX}[row sep=5mm]
							\zxN{} \arrow[d] & [5mm] & [5mm] \zxX{} \arrow[d]\\
							\zxN{} \arrow[dr,bend right] & & 
							\zxN{} \arrow[dl,bend left]\\
							& \zxX{} \arrow[d] &\\
							& \zxN{} &
						\end{ZX}\qquad\qquad\text{Unitality}
					\end{equation*}
					
					\item[ii.)] a comultiplication and counit
					\begin{equation*}
						\Delta_{\textcolor{green}{\bullet}}:=\begin{ZX}[row sep=5mm]
							& [5mm] \zxN{} \arrow[d,] & [5mm]\\
							& \zxZ{} \arrow[dl,bend right] \arrow[dr,bend left] & \\
							\zxN{} & & \zxN{}
						\end{ZX}\quad\colon A\to A\otimes A
						\qquad\qquad
						\varepsilon_{\textcolor{green}{\bullet}}:=   \begin{ZX}[row sep=5mm]
							\zxN{} \arrow[d]\\
							\zxZ{}
						\end{ZX}\quad\colon A\to\Bbbk    
					\end{equation*}
					satisfying
					\begin{equation*}
						\begin{ZX}[row sep=5mm]
							& [5mm] & [5mm] \zxN{} \arrow[d,] & [5mm]\\
							& & \zxZ{} \arrow[dl,bend right] \arrow[dr,bend left] & \\
							& \zxN{} \arrow[d] & & \zxN{} \arrow[dd]\\
							& \zxZ{} \arrow[dl,bend right] \arrow[dr,bend left] & &\\
							\zxN{} & & \zxN{} & \zxN{}
						\end{ZX}\quad=\quad
						\begin{ZX}[row sep=5mm]
							& [5mm] \zxN{} \arrow[d,] & [5mm] & [5mm]\\
							& \zxZ{} \arrow[dl,bend right] \arrow[dr,bend left] & & \\
							\zxN{} \arrow[dd] & & \zxN{} \arrow[d] &\\
							& & \zxZ{} \arrow[dl,bend right] \arrow[dr,bend left] & \\
							\zxN{} & \zxN{} & & \zxN{}
						\end{ZX}\qquad\qquad\text{Coassociativity}
					\end{equation*}
					and
					\begin{equation*}
						\begin{ZX}[row sep=5mm]
							& [5mm] \zxN{} \arrow[d,] & [5mm]\\
							& \zxZ{} \arrow[dl,bend right] \arrow[dr,bend left] & \\
							\zxN{} \arrow[d] & & \zxN{} \arrow[d]\\
							\zxZ{} & & \zxN{}
						\end{ZX}\quad=\quad
						\begin{ZX}[row sep=5mm]
							\zxN{} \arrow[ddd]\\
							\\
							\\
							\zxN{}
						\end{ZX}
						\quad=\quad\begin{ZX}[row sep=5mm]
							& [5mm] \zxN{} \arrow[d,] & [5mm]\\
							& \zxZ{} \arrow[dl,bend right] \arrow[dr,bend left] & \\
							\zxN{} \arrow[d] & & \zxN{} \arrow[d]\\
							\zxN{} & & \zxZ{}
						\end{ZX}\qquad\qquad\text{Counitality}    
					\end{equation*}
					
					\item[iii.)] a normalization constant $\lozenge\in\Bbbk\setminus\{0\}$ such that $\Delta$ and $\varepsilon$ are unnormalized algebra morphisms
					\begin{equation*}
						\lozenge\begin{ZX}[row sep=5mm]
							& [5mm] \zxX{} \arrow[d] & [5mm]\\
							& \zxZ{} \arrow[dl,bend right] \arrow[dr,bend left] & \\
							\zxN{} & & \zxN{}
						\end{ZX}\quad=\quad
						\begin{ZX}[row sep=5mm]
							\zxX{} \arrow[dd] & [5mm] \zxX{} \arrow[dd]\\
							& \\
							\zxN{} & \zxN{}
						\end{ZX}\qquad\qquad
						\lozenge\begin{ZX}[row sep=5mm]
							& [5mm] \zxN{} \arrow[d] & [5mm] & [5mm] \zxN{} \arrow[d] & [5mm]\\
							& \zxZ{} \ar[dd,IO,C] \arrow[s, ddrr] & & \zxZ{} \arrow[s, ddll] \ar[dd,IO,C-] &\\
							& & \zxN{} & & \\
							& \zxX{} \arrow[d] & & \zxX{} \arrow[d] &\\
							& \zxN{} & & \zxN{} &
						\end{ZX}\quad=\quad
						\begin{ZX}[row sep=5mm]
							\zxN{} \arrow[dr,bend right] & [5mm] & [5mm] \zxN{} \arrow[dl,bend left]\\
							& \zxX{} \arrow[d] &\\
							& \zxZ{} \arrow[dr,bend left] \arrow[dl,bend right] &\\
							\zxN{} & & \zxN{}
						\end{ZX}
					\end{equation*}
					~\\
					\begin{equation*}
						\begin{ZX}[row sep=5mm]
							\zxX{} \arrow[d]\\
							\zxZ{}
						\end{ZX}\quad=\quad
						\lozenge\qquad\qquad
						\lozenge\begin{ZX}[row sep=5mm]
							\zxN{} \arrow[dr,bend right] & [5mm] & [5mm] \zxN{} \arrow[dl,bend left]\\
							& \zxX{} \arrow[d] & \\
							& \zxZ{} &
						\end{ZX}\quad=\quad
						\begin{ZX}[row sep=5mm]
							\zxN{} \arrow[dd] & [5mm] \zxN{} \arrow[dd]\\
							& \\
							\zxZ{} & \zxZ{}
						\end{ZX}
					\end{equation*}
					
					\item[iv.)] and a linear map $S:=\begin{ZX}
						\zxN{} \arrow[d]\\
						\zxX{S} \arrow[d]\\
						\zxN{}
					\end{ZX}\colon A\to A$ satisfying the antipode axioms
					\begin{equation}\label{antipode}
                        \lozenge\quad\lozenge\quad
						\begin{ZX}[row sep=5mm]
							& [5mm] \zxN{} \arrow[d] & [5mm]\\
							& \zxZ{} \arrow[dl,bend right] \ar[dd,IO,C-] & \\
							\zxX{S} \arrow[dr,bend right] & &\\
							& \zxX{} \arrow[d] &\\
							& \zxN{} &
						\end{ZX}\quad=\quad
						\begin{ZX}[row sep=5mm]
							\zxN{} \arrow[d]\\
							\zxZ{}\\
							\zxX{} \arrow[d]\\
							\zxN{}
						\end{ZX}\quad=\quad
                            \lozenge\quad\lozenge
						\begin{ZX}[row sep=5mm]
							& [5mm] \zxN{} \arrow[d] & [5mm]\\
							& \zxZ{} \arrow[dr,bend left] \ar[dd,IO,C] & \\
							& & \zxX{S} \arrow[dl,bend left]\\
							& \zxX{} \arrow[d] &\\
							& \zxN{} &
						\end{ZX}
					\end{equation}
				\end{enumerate}
				The usual notion of (normalized) Hopf algebra is recovered by setting $\lozenge=1$. It is easy to see that, given an unnormalized Hopf algebra $(A,m_{\textcolor{red}{\bullet}},\eta_{\textcolor{red}{\bullet}},\Delta_{\textcolor{green}{\bullet}},\varepsilon_{\textcolor{green}{\bullet}},S)$, we obtain a normalized Hopf algebra $(A,m_{\textcolor{red}{\bullet}},\eta_{\textcolor{red}{\bullet}},\Delta'_{\textcolor{green}{\bullet}},\varepsilon'_{\textcolor{green}{\bullet}},S')$ by setting $\Delta'_{\textcolor{green}{\bullet}}:=\lozenge\Delta_{\textcolor{green}{\bullet}}$, $\varepsilon'_{\textcolor{green}{\bullet}}:=\lozenge^{-1}\varepsilon_{\textcolor{green}{\bullet}}$ and $S':=\lozenge^{-2}S$.
				
				\begin{example}\label{ExKG}
					Let $G$ be a group with neutral element $e\in G$. Then the group algebra $\Bbbk G$ becomes a (normalized) Hopf algebra with algebra structure $m_{\textcolor{red}{\bullet}}(g\otimes h):=gh$, $\eta_{\textcolor{red}{\bullet}}(1_\Bbbk):=e$, coalgebra structure $\Delta_{\textcolor{green}{\bullet}}(g)=g\otimes g$, $\varepsilon_{\textcolor{green}{\bullet}}(g)=1$ and antipode $S(g)=g^{-1}$ defined for all $g,h\in G$ and extended $\Bbbk$-linearly to $\Bbbk G$.
				\end{example}

				\subsection{\texorpdfstring{$F$}{F}-Hopf algebras}\label{appFHopf}
				
				\begin{definition}\label{defFAlg}
				A vector space $A$ is called \textit{$F$-algebra} if there is an associative unital algebra structure $(m=\begin{ZX}
						\zxN{} \arrow[dr, bend right] & & \zxN{} \arrow[dl, bend left]\\
						& \zxN{} \arrow[d] & \\
						& \zxN{}
						&\end{ZX},\eta)$ on $A$ and a coassociative counital coalgebra structure $(\Delta=
					\begin{ZX}[]
						& \zxN{} \arrow[d] & \\
						& \zxN{} \arrow[dl, bend right] \arrow[dr, bend left] & \\
						\zxN{} & & \zxN{}
					\end{ZX},\varepsilon)$ on $A$ such that
					\begin{equation*}
						\begin{ZX}[]
							& \zxN{} \arrow[d] & & & \zxN{} \arrow[d]\\
							& \zxN{} \arrow[dl, bend right] \arrow[dr, bend left] & & & \zxN{} \arrow[d]\\
							\zxN{} \arrow[d] & & \zxN{} \arrow[dr, bend right] & & \zxN{} \arrow[dl, bend left]\\
							\zxN{} \arrow[d] & & & \zxN{} \arrow[d] &\\
							\zxN{} & & & \zxN{} &
						\end{ZX}\quad=\quad
						\begin{ZX}[]
							\zxN{} \arrow[dr, bend right] & & \zxN{} \arrow[dl, bend left]\\
							& \zxN{} \arrow[dd] & \\
							& & \\
							& \zxN{} \arrow[dl, bend right] \arrow[dr, bend left] & \\
							\zxN{} & & \zxN{}
						\end{ZX}
						\quad=\quad
						\begin{ZX}[]
							\zxN{} \arrow[d] & & & \zxN{} \arrow[d] & \\
							\zxN{} \arrow[d] & & & \zxN{} \arrow[dl, bend right] \arrow[dr, bend left] & \\
							\zxN{} \arrow[dr, bend right] & & \zxN{} \arrow[dl, bend left] & & \zxN{} \arrow[d]\\
							& \zxN{} \arrow[d] & & & \zxN{} \arrow[d]\\
							& \zxN{} & & & \zxN{}
						\end{ZX}
					\end{equation*}
					hold.
				\end{definition}
				\begin{definition}
					A pair
					$$
					(A,m_{\textcolor{red}{\bullet}},\eta_{\textcolor{red}{\bullet}},\Delta_{\textcolor{red}{\bullet}},\varepsilon_{\textcolor{red}{\bullet}})\qquad,\qquad
					(A,m_{\textcolor{green}{\bullet}},\eta_{\textcolor{green}{\bullet}},\Delta_{\textcolor{green}{\bullet}},\varepsilon_{\textcolor{green}{\bullet}})
					$$
					of $F$-algebras on the same vector space $A$ is called an \textit{$F$-bialgebra} if
					$$
					(A,m_{\textcolor{red}{\bullet}},\eta_{\textcolor{red}{\bullet}},\Delta_{\textcolor{green}{\bullet}},\varepsilon_{\textcolor{green}{\bullet}})\qquad,\qquad
					(A,m_{\textcolor{green}{\bullet}},\eta_{\textcolor{green}{\bullet}},\Delta_{\textcolor{red}{\bullet}},\varepsilon_{\textcolor{red}{\bullet}})
					$$
					are bialgebras. The pair is called an \textit{$F$-Hopf algebra} if the latter are Hopf algebras. 
				\end{definition}
				
				\begin{example}
					Let $G$ be a finite group with neutral element $e\in G$. Then $A=\Bbbk G$ is an $F$-Hopf algebra with
					$$
					m_{\textcolor{red}{\bullet}}(g\otimes h)=gh,\qquad
					\eta_{\textcolor{red}{\bullet}}(1_\Bbbk)=e,\qquad
					\Delta_{\textcolor{red}{\bullet}}(g)=\sum_{h,l\in G, hl=g}h\otimes l,\qquad
					\varepsilon_{\textcolor{red}{\bullet}}(g)=\delta_{e,g}
					$$
					and
					$$
					m_{\textcolor{green}{\bullet}}(g\otimes h)=\delta_{g,h}g,\qquad
					\eta_{\textcolor{green}{\bullet}}(1_\Bbbk)=\sum_{g\in G}g,\qquad
					\Delta_{\textcolor{green}{\bullet}}(g)=g\otimes g,\qquad
					\varepsilon_{\textcolor{green}{\bullet}}(g)=1.
					$$
					Namely, $(A,m_{\textcolor{red}{\bullet}},\eta_{\textcolor{red}{\bullet}},\Delta_{\textcolor{green}{\bullet}},\varepsilon_{\textcolor{green}{\bullet}})$ is the Hopf algebra from Example \ref{ExKG} with antipode $S(g)=g^{-1}$, while $(A,m_{\textcolor{green}{\bullet}},\eta_{\textcolor{green}{\bullet}},\Delta_{\textcolor{red}{\bullet}},\varepsilon_{\textcolor{red}{\bullet}})$ is a bialgebra because, using the notation $m_{\textcolor{green}{\bullet}}(g\otimes h)=:g\textcolor{green}{\bullet}h$,
					\begin{align*}
						\Delta_{\textcolor{red}{\bullet}}(g)\textcolor{green}{\bullet}\Delta_{\textcolor{red}{\bullet}}(h)
						&=\sum_{k,k',n,n'\in G, kn=g,k'n'=h}k\textcolor{green}{\bullet}k'\otimes l\textcolor{green}{\bullet}l'\\
						&=\sum_{k,l\in G, kl=g}k\otimes l\\
						&=\Delta_{\textcolor{red}{\bullet}}(g\textcolor{green}{\bullet}h)
					\end{align*}
					and
					$$
					\varepsilon_{\textcolor{red}{\bullet}}(g)\varepsilon_{\textcolor{red}{\bullet}}(h)
					=\delta_{e,g}\delta_{e,h}
					=\delta_{e,g}\delta_{g,h}
					=\varepsilon_{\textcolor{red}{\bullet}}(g\textcolor{green}{\bullet}h).
					$$
					follow. The latter is endowed with the antipode $S(g)=g^{-1}$.
				\end{example}
				The interesting example of quantum $sl(2)$ can be found in \cite[Ex. 2.11]{MajidZX}.


\begin{thebibliography}{99}

\bibitem{amari}
Shun-Ichi Amari, \textit{Natural gradient works efficiently in learning}. Neural
computation, 10(2):251–276, 1998.

\bibitem{as}
Abhay Ashtekar and Troy A. Schilling,
\textit{Geometrical Formulation of Quantum Mechanics}.
In “On Einstein's Path”, A. Harvey Ed., Springer, 1999.


\bibitem{shor3} Adriano Barenco, Charles H. Bennett, Richard Cleve, David P. DiVincenzo,
Norman Margolus, Peter Shor, Tycho Sleator, John Smolin and Harald Weinfurther,
\textit{Elementary gates for quantum computation}.
Physical Review A, March 22, 1995.


\bibitem{bkw} Niel de Beaudrap, Aleks Kissinger, and John van de Wetering, \textit{Circuit Extraction
for ZX-diagrams can be \# P-hard}. 
In 49th International Colloquium on Automata, Languages, and Programming (ICALP 2022). LIPIcs 229, 119:1-119:19, Schloss Dagstuhl - Leibniz-Zentrum für Informatik, 2022.


\bibitem{benioff}
Paul Benioff, \textit{The computer as a physical system: A microscopic quantum
mechanical Hamiltonian model of computers as represented by turing machines.}
Journal of Statistical Physics, 22(5):563–591, 1980.

\bibitem{cd1} Bob Coecke and Ross Duncan, \textit{Interacting quantum observables.} In Proceedings
of the 37th International Colloquium on Automata, Languages and Programming
(ICALP), Lecture Notes in Computer Science, 2008. 

\bibitem{CoeckeDuncan} Bob Coecke and Ross Duncan, 
\textit{Interacting quantum observables: categorical algebra
and diagrammatics.} New Journal of Physics, 13:043016, 2011.


\bibitem{dklp} Eric Dennis, Alexei Kitaev, Andrew Landahl and John Preskill,
\textit{Topological quantum
memory}. Journal of Mathematical Physics, 43(9):4452–4505, 2002.


\bibitem{deutsch}
David Deutsch, \textit{Quantum theory, the Church–Turing principle and the
universal quantum computer.}
Proceedings of the Royal Society of London. A. Mathematical and
Physical Sciences, 400(1818):97–117, 1985.

\bibitem{fkm}
Paolo Facchi, Ravi Kulkarni, V.I. Man’ko, Giuseppe Marmo, E.C.G. Sudarshan and Franco Ventriglia,
\textit{Classical and quantum Fisher information in the geometrical formulation of quantum mechanics}.
Physics Letters A 374 (2010) 4801-4803.

\bibitem{feynman} Richard P. Feynman, \textit{Simulating physics with computers}.
Int. J. Theor. Phys., 21 (1982).

\bibitem{feynman1987} Richard P. Feynman, \textit{Quantum mechanical computers}.
  Optics News 11, 11-20. Also in Foundations of Physics, 16(6), 507-531, 1986.

\bibitem{fz}
Rita Fioresi and Ferdinando Zanchetta \textit{Deep learning and geometric deep learning: An introduction for mathe-
maticians and physicists}, International Journal of Geometric Methods in Modern Physics, 2023.


\bibitem{grover} Lov K. Grover, \textit{Quantum mechanics helps in searching
for a needle in a haystack}. Phys, Rev, Let., 79, 1997.

\bibitem{helstrom} Carl W. Helstrom,
\textit{Quantum detection and estimation theory}. Academic Press, 1976.


\bibitem{hou} Xu-Yang Hou, Zheng Zhou, Xin Wang, Hao Guo and Chih-Chun Chien,
\textit{Local geometry and quantum geometric tensor of mixed states}. Preprint arXiv:2305.07597.

\bibitem{huy} Daniel Huybrechts, \textit{Complex Geometry: An Introduction}. Springer GTM, 2005.

\bibitem{kassel}
Christian Kassel,
\textit{Quantum Groups.} Springer-Verlag, GTM, 1995.

\bibitem{kitaev1} Alexei Yu. Kitaev, \textit{Fault-tolerant quantum computation
by anyons}. Ann. Phys., 303(1), 2-30, 2003.

\bibitem{kitaev2} Alexei Yu. Kitaev, \textit{Periodic table for topological insulators
and superconductors}. AIP Conference Proceedings 1134, 22, 2009.


\bibitem{liu} Jing Liu, Xiaoxing Jing, Wei Zhong and Xiaoguang Wang,
\textit{Quantum Fisher information for density matrices with arbitrary ranks}.
Communications in Theoretical Physics, Vol. 61, Issue 1, 2014.

\bibitem{MajidZX}
Shahn Majid,
\textit{Quantum and braided ZX calculus}. J. Phys. A: Math. Theor. \textbf{55} 254007, 2022.

\bibitem{majid-found}
Shahn Majid,
\textit{Foundations of Quantum Groups Theory}. Cambridge University Press, 1995.

\bibitem{manin} Yuri Manin, 
  \textit{Computable and Uncomputable}. Sovetskoye Radio, Moscow, 128, 1980.



\bibitem{manin2} Yuri Manin, \textit{Topics in noncommutative geometry}. M. B. Porter Lectures,
Princeton University Press, Princeton, NJ, 1991.


\bibitem{montgomery} Susan Montgomery, \textit{Hopf Algebras and Their Actions on Rings}.
CBMS Lecture Notes vol 82, American Math Society, Providence, RI, 1993.

\bibitem{preskill1}
John Preskill, \textit{Reliable quantum computers.}
Proceedings of the Royal Society of London.
Series A: Mathematical, Physical and Engineering Sciences, 454(1969):385–410,
1998.

\bibitem{preskill2} John Preskill,
\textit{Fault-tolerant quantum computation}. In Introduction to quantum
computation and information, 213–269. World Scientific, 1998.

\bibitem{preskill3} John Preskill,
  \textit{Lecture Notes for Physics 229: Quantum Information and Computation}.
  Barron, 2015.

\bibitem{shor1} Peter W. Shor, \textit{Algorithms for quantum computation, discrete logarithms
and factorizing}. Proceedings 35th annual symposium on foundations of com-
puter science, IEEE Press, 1994.

\bibitem{shor2} Peter W. Shor, \textit{
Scheme for reducing decoherence in quantum computer memory.}
Physical Review A, 52(4):R2493, 1995.



\bibitem{sikc}
James Stokes, Josh Izaac, Nathan Killoran and Giuseppe Carleo,
\textit{Quantum Natural Gradient}. Quantum 4, 269, 2020.


\bibitem{zxcalc} John van de Wetering,
  \textit{ZX-calculus for the working quantum computer scientist}. 2020,
  {https://api.semanticscholar.org/CorpusID:229680014}


\end{thebibliography}
\end{document}